\documentclass[english]{article}
\usepackage[T1]{fontenc}
\usepackage[utf8]{inputenc}
\usepackage{geometry}
\usepackage{color}
\usepackage{babel}
\usepackage{amsmath}
\usepackage{amsthm}
\usepackage{graphicx}
\usepackage{setspace}
\usepackage[unicode=true]
 {hyperref}
\usepackage{breakurl}

\makeatletter
\theoremstyle{plain}
\newtheorem{ax}{\protect\axiomname}
\theoremstyle{plain}
\newtheorem{thm}{\protect\theoremname}

\makeatother

\providecommand{\axiomname}{Axiom}
\providecommand{\theoremname}{Theorem}

\begin{document}
\begin{doublespace}
\begin{center}
\textbf{\large{}For Whom the Bell (Curve) Tolls: A to F, Trade Your
Grade Based on the Net Present Value of Friendships with Financial
Incentives}{\large\par}
\par\end{center}

\begin{center}
\textbf{Ravi Kashyap }
\par\end{center}

\begin{center}
\textbf{SolBridge International School of Business / City University
of Hong Kong }
\par\end{center}

\begin{center}
Oct-25-2017
\par\end{center}

\begin{center}
\par\end{center}

\begin{center}
JEL Codes: G1 General Financial Markets; G2 Financial Institutions
and Services; D81 Criteria for Decision-Making under Risk and Uncertainty
\par\end{center}

\begin{center}
\textbf{\textcolor{blue}{\href{https://doi.org/10.3905/jpe.2019.22.3.064}{Edited Version: Kashyap, R. (2019). For Whom the Bell (Curve) Tolls: A to F, Trade Your Grade Based on the Net Present Value of Friendships with Financial Incentives. The Journal of Private Equity, 22(3), 64-81. }}}
\par\end{center}

\begin{center}
\tableofcontents{}\pagebreak{}
\par\end{center}
\end{doublespace}
\begin{doublespace}

\section{Abstract}
\end{doublespace}

\begin{doublespace}
\noindent We discuss a possible solution to an unintended consequence
of having grades, certificates, rankings and other diversions in the
act of transferring knowledge; and zoom in specifically to the topic
of having grades, on a curve. We conduct a thought experiment, taking
a chapter (and some more?) from the financial markets, (where we trade
pollution and what not?), to create a marketplace, where we can trade
our grade, similar in structure to the interest rate swap. We connect
this to broader problems that are creeping up, unintentionally, due
to artificial labels we are attaching, to ourselves. The policy and
philosophical implications of our arguments are to suggest that all
trophies that we collect (including certificates, grades, medals etc.)
should be viewed as personal equity or private equity (borrowing another
widely used term in finance) and we should not use them to determine
the outcomes in any selection criteria except have a cutoff point:
either for jobs, higher studies, or, financial scholarships, other
than for entertainment or spectator sports. We suggest alternate methods
for grading and performance assessment and put forth tests for teaching
and learning similar to the Turing Test for intelligence.
\end{doublespace}
\begin{doublespace}

\section{\label{sec:Good-Intentions,-Bad}Good Intentions, Bad Consequences
... Ugly Repercussions}
\end{doublespace}

\begin{doublespace}
\noindent Stating that we have all been students at some point in
our lives would be an assumption not far from reality. Here, we are
referring to students in the currently understood sense of the word,
as of the 21st century, who wish to get good grades since we seem
to think that there is a high correlation between grades and employment
opportunities, scholarships, further academic pursuits, financial
assistance and many other benefits (Jones \& Jackson 1990; Loury \&
Garman 1995; Rumberger \& Thomas 1993; French, Homer, Popovici \&
Robins 2015; Daley \& Green 2014; End-note \ref{Student}).

\noindent We discuss a possible solution to an unintended consequence
of having grades (and related artifacts such as college diplomas or
course completion certificates), which is that students are more focused
on grades rather than on learning, which is the true purpose of being
students. Another unintended consequence which has immediate repercussions,
to both students and educational institutions is that certain students
who are receiving financial scholarships might end up losing their
funding if they fail to obtain a certain minimum level of grades:
(Johnstone 2004) is a discussion of the shift in at least part of
the higher educational cost burden from governments, or taxpayers,
to parents and students, which has necessitated that students look
out for different financial assitance schemes to provide for their
higher education costs; (Monks 2009) finds that merit aid has a statistically
significant but inelastic effect on enrollment of extremely high ability
students; (Henry, Rubenstein \& Bugler 2004) study students just above
the eligibility threshold for assistance and find that losing scholarships
may substantially reduce any potential positive effects of receiving
it in the first place, suggesting that efforts to increase the number
of students who retain the scholarships should be a major focus of
future policy initiatives.

\noindent (Carruthers \& Özek 2016) find that losing one’s scholarship
results in a small degree of detachment from college, though thankfully,
no effect on timely completion in majority of the cases; (Gross, Hossler,
Ziskin \& Berry 2015) find no relationship between institutional merit
aid and student departure, instead need-based aid was consistently
related to decreased odds of departure; (Zhang \& Ness 2010) find
that state merit scholarship programs do indeed stanch the migration
of “best and brightest” students to other states; (Sjoquist \& Winters
2015) find that exposure to state merit aid programs have no meaningfully
positive effect on college completion; also see: (Cohen-Vogel, Ingle,
Levine \& Spence 2008; Henry \& Rubenstein 2002; Cornwell, Mustard
\& Sridhar 2006 for the pros, cons and other aspects of merit based
financial aid; End-note \ref{Student Financial Aid}).

\noindent \textbf{\textit{Giving scholarships based on grades on a
curve can be one way to provide an incentive for students to maintain
high scores (significant debate exists on the use of grades, with
or without a curve, as incentives to learn: }}\textbf{Stan 2012; Betts
\& Grogger 2003; Kulick \& Wright 2008}\textbf{\textit{; }}\textbf{Grant
\& Green 2013}\textbf{\textit{; Figlio \& Lucas 2004). As long as
students are not too caught up on their grades, (easier said than
done, but we discuss ways that could accomplish this in later sections)
scholarships built upon a grade curve can be a decent decision making
criteria of handing out limited resources, which in this case are
monetary funds available, to the large body of students that seek
assistance. Also, grading on a curve could be one way mitigate the
effects of grade inflation (Chan, Hao \& Suen2007; Johnson 2006; Eiszler
2002; Kohn 2002; Sabot \& Wakeman-Linn 1991; Zangenehzadeh 1988; End-note
\ref{enu:Grade-inflation-is}).}}

\noindent There has been significant debate regarding the usage of
the Bell Curve in societal contexts: Murray, C., \& Herrnstein 1994
is a highly controversial book titled the Bell Curve. The book's title
comes from the bell-shaped normal distribution of intelligence quotient
(IQ) scores in a population; the authors argue that human intelligence
is substantially influenced by both inherited and environmental factors
and that it is a better predictor of many personal dynamics, including
financial income, job performance, birth out of wedlock, and involvement
in crime than are an individual's parental socioeconomic status. Kincheloe,
Steinberg \& Gresson 1997 is a collection of essays criticizing this
book. Sternberg 1995; Ma \& Schapira 2017 are other reviews on this
topic that highlight the fact that the factors affecting IQ are still
poorly understood. Beardsley 1995; Jacoby \& Glauberman 1995; Fendler
\& Muzaffar 2008 consider the historical aspects of the bell curve,
both the concept and the above mentioned book, and highlight the possibility
that history is usually intertwined with politics. End-note \ref{enu:An-intelligence-quotient}.

\noindent 
\begin{figure}
\includegraphics[width=17.5cm]{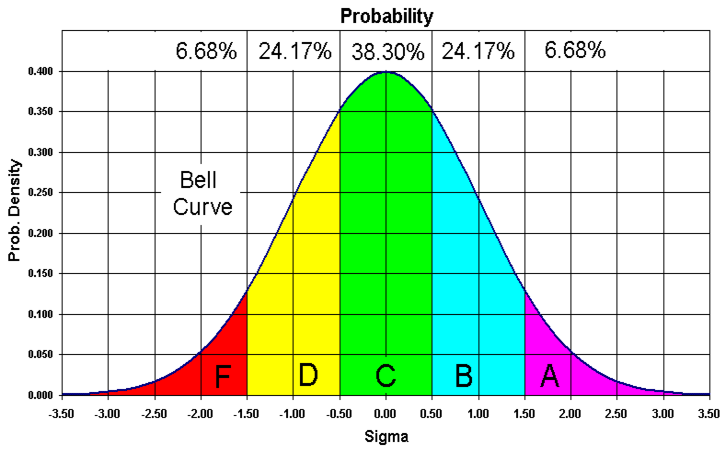}

\caption{\label{fig:A-to-F:}A to F: For Whom the Bell (Curve) Tolls ...}
\end{figure}

\noindent The issue of students losing assistance is exacerbated (or,
perhaps, even caused) due to guidelines (more of a law, since they
need to be followed), which prescribe that students must be graded
on a curve (Fendler \& Muzaffar 2008; Taylor 1971; Reeves 2001; End-note
\ref{Grade Curve}; Figure \ref{fig:A-to-F:}), which limits the number
of students that can obtain top grades. (Durm 1993) is a summary of
the history of how grading policies developed and the assumptions
that might have led to academe creating its own nightmare; this is
covered in greater detail in (Finkelstein 1913; Smallwood 1969; Cureton
1971; Schneider \& Hutt 2014). There is a good amount of research
on grading policies, whether it is a good practice, potential obstacles
to reform, and other aspects of classroom performance and related
consequences: (Zeidner 1992; Polloway, Epstein, Bursuck, Roderique,
McConeghy \& Jayanthi 1994; Stiggins, Frisbie \& Griswold 1989; Pintrich
\& De Groot, 1990; Harackiewicz, Barron, Carter, Lehto \& Elliot 1997;
Barron \& Harackiewicz 2003; Tomlinson 2005; Guskey 1994; 2011).

\noindent \textbf{\textit{We want to emphasize, that these are unintended
consequences, since the establishment of universities, educational
foundations, introduction of grades, or, having grades on a curve
are wonderful innovations, done with honorable intentions, that might
have been fabricated to solve certain other problems. One reason,
why such unwanted outcomes creep up, is because, we live in a world
that requires around 2000 IQ points, to consistently make correct
decisions; but the smartest of us has only a fraction of that (Ismail
2014; End-note \ref{enu:Taleb and Kahneman discuss Trial and Error / IQ Points}).
Hence, we need to rise, above the urge to ridicule, the seemingly
obvious blunders of others, since without those marvelous mistakes,
the path ahead will not become clearer for us.}}

\noindent Another more unwelcome unintended consequence of having
grades on a curve, is perhaps that students end up competing, with
one another (Bell, Grekul, Lamba, Minas \& Harrell 1995 found that
students were most likely to help when the other students were friends,
when there was frequent contact, and when the test was not graded
on a curve, that is absence of competition; Schneider \& Hutt 2014
is a historical interpretation, origins, uses and evolution, of grades;
also see: Schinske \& Tanner 2014; Brookhart etal 2016; Bresee 1976;
Natriello 1987; Aviles 2001). While a competitive spirit, is to be
encouraged; an overtly competitive environment destroys the spirit
of co-operation among the members of the classroom, who are all travelers
on the same journey of learning since they learn not only from the
instructors, but perhaps, more from one another.

\noindent We conduct a thought experiment, taking a chapter (and some
more?) from the financial markets, (where we trade pollution and what
not?), to create a marketplace, where we can trade our grade, similar
in structure to the interest rate swap (section \ref{sec:The-Grade-Trade}).
We want to emphasize that this experiment is meant to be merely hypothetical
at this stage. Most importantly we highlight the practical obstacles
and many legal and moral implications of even considering such an
approach. The benefit of this thought excercise is that it possibly
provides solutions to the unintended consequences of having grades,
certificates, rankings and other diversions in the act of transferring
knowledge; and zoom in specifically to the topic of having grades,
on a curve. We connect our discussion to broader problems that are
creeping up, unintentionally, due to artificial labels we are attaching,
to ourselves.

\noindent The policy and philosophical implications of our arguments
are to suggest that all trophies that we collect (including certificates,
grades, medals etc.) should be viewed as personal equity or private
equity (borrowing another widely used term in finance) and we should
not use them to determine the outcomes in any selection criteria except
have a cutoff point: either for jobs, higher studies, or, financial
scholarships, other than for entertainment or spectator sports. Given
our fascination with testing and measurement, we suggest alternate
methods for grading (sections \ref{sec:Infinite-Progress-Benchmark},
\ref{subsec:Risk-Management-of}) and put forth tests for teaching
and learning similar to the Turing Test for intelligence (section
\ref{sec:Turing-Tests-for}).

\noindent Lastly, it would not be entirely incorrect to state that
the majority of the attempts at evaluating students, and perhaps even
most (all?) of knowledge creation, starts with answering questions.
In present day society, we seem to be focused on answering questions
that originate in different disciplines. Hence, as a first step, we
recognize that one possible categorization of different fields can
be done by the set of questions a particular field attempts to answer.
Since we are the creators of different disciplines, but not the creators
of the world in which these fields need to operate (based on our present
understanding of our role in the cosmos, which might very well change),
the answers to the questions posed by any domain can come from anywhere
or from phenomenon studied under a combination of many other disciplines.

\noindent Hence, the answers to the questions posed under the realm
of education or knowledge transference can come from seemingly diverse
subjects, such as, physics, biology, mathematics, chemistry, marketing,
economics, finance and so on. This suggests that we might be better
off identifying ourselves with problems and solutions, which tacitly
confers upon us the title Problem Solvers, instead of calling ourselves
teachers, professors, physicists, biologists, psychologists, marketing
experts, economists and so on. This quest for answers is bounded only
by our imagination (Calaprice 2000). Our paper then becomes an example
of how concepts from finance and trading, can be used to solve problems
in education with many connections to deeper dilemmas in society.
\end{doublespace}
\begin{doublespace}

\section{\label{sec:The-Grade-Trade}The Grade Trade in a Light Pool}
\end{doublespace}

\begin{doublespace}
\noindent Our innovation is to use the financial markets to facilitate
trades on grades (End-note \ref{Grade-Trade-Market}). Students that
have good scores or higher grades, can trade their grade with someone
that might end up losing their scholarships, or other opportunities,
because of missing out the grades they need. We call this marketplace
for grades, ``A Light Pool'', in contrast to dark pools used for
financial instruments since we would aim for complete transparency
and also because knowledge represents light (Mittal 2008; Domowitz,
Finkelshteyn \& Yegerman 2008; Ganchev, Nevmyvaka, Kearns \& Vaughan
2010; Zhu 2014 explain dark pools in finance and their potential benefits
and possible harmful aspects; End-note \ref{Dark Pool}; Kvanvig 2003;
Pritchard 2009; Kashyap 2017c are discussions about  the value of
knowledge and the long held beliefs regarding higher powers in the
knowledge realm; also see: Turner \& Coulter 2001; Hallam 1996; Ludvik
2007; End-notes \ref{Knowledge Deities}, \ref{Knowledge}; Dancy,
Sosa \& Steup 2009; DeRose 2005; Figueroa 2016; Pritchard 2018; Hetherington
2018; End-note \ref{Epistemology} are an excellent collection of
articles on leading theories, thinkers, ideas, distinctions, central
questions and concepts in epistemology).

\noindent The person who is giving up the higher grades, in return,
benefits from the following items: 1) They earn the good will and
friendship of the person who would have lost out a much needed source
of money (Berndt 2002 discuss the relationship between the quality
of friendships and social development; End-note \ref{Friend Indeed}).
2) They can get a certain percentage of the funds that might have
been lost without the trade. 3) The grade swap will be reversed when
there is no benefit to the person receiving the higher grade in the
original trade; perhaps even with another financial component as agreed
upon initially when the first trade is made. The structure, which
is discussed in detail below, is given in Figure \ref{fig:Grade-Trade:-Similar}.

\noindent 
\begin{figure}
\includegraphics[width=17.5cm]{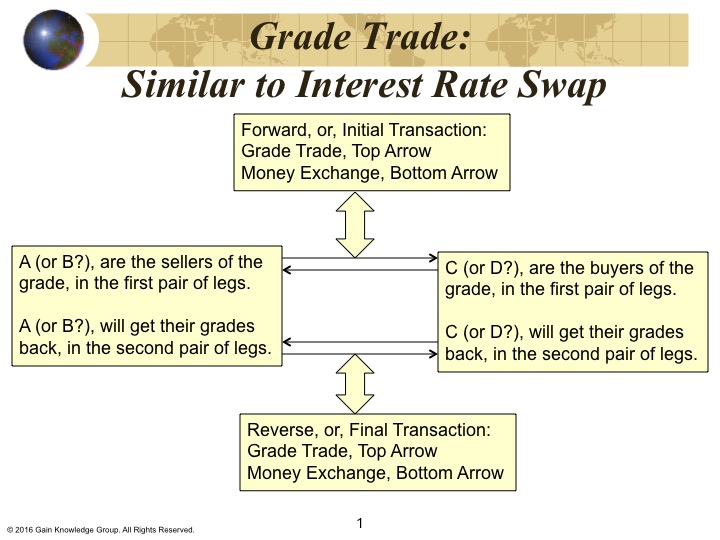}

\caption{\label{fig:Grade-Trade:-Similar}Grade Trade: Similar to Interest
Rate Swap}

\end{figure}

\noindent To ensure fairness, we set a guideline, based on a recent
example (as of October 25, 2017) set by Facebook founder Mark Zuckerberg,
who has pledged, on Dec 1, 2015, to give away 99\% of his wealth to
advancing human potential (Schervish, Davis, Cosnotti \& Rosplock
2016; Chua, Aricat \& Goh 2017; End-note \ref{Zuckberg 99=000025, Human Potential Facebook Link},
\ref{Zuckerberg 99=000025 News Links}); technically, 99\% of his
Facebook shares will be put into a legal entity, a limited liability
company and not a charity, which would be owned and controlled by
Zuckerberg. As a private company, this initiative can spend its money
on whatever it wants, including private, profit-generating investment
(Aron 2016; Clark \& McGoey 2016; End-note \ref{Zuckerberg 99=000025 in Legal Entity}).
His net-worth at that time was around, \$45 Billion USD. Hence, we
suggest that the trade for a grade, should involve a financial component,
that is less than 1\% of the scholarship amount. Surely, if someone
with \$45 Billion USD can give away 99\% of it. Then anyone, with
much less than \$45 Billion, (we do not recollect any financial scholarship
to any institution, let alone to any individual student, being more
than or equal to this amount) can afford to give away something much
less than that one percentage amount.

\noindent We would again like to highlight that, this less than 1\%
suggestion is not to poke fun at Mark Zuckerberg for keeping a few
hundred million before giving out everything else. He has every right
to keep everything he has earned. His earnings are surely much lesser
in value compared to the wonderful legacy he has established, which
connects most of humanity, as we know it, with one other. We merely
use this example to come up with a number that can act as a upper
bound for the financial component to this grade trade. 

\noindent We further provide theoretical and mathematical justifications
based on three axioms (End-note \ref{enu:The-reason-Rationalizations}):
\end{doublespace}
\begin{enumerate}
\begin{doublespace}
\item The value of grades generally decline, as time passes (Figure \ref{fig:Exponential-Decay:-Time}).
\item The value of friendships usually rallies, as time passes (Figure \ref{fig:Exponential-Growth:-Time}).
\item The value of money, can increase, as time passes, even after considering
inflation, provided, judicious investments are made (Figure \ref{fig:Exponential-Growth:-Time}).
\end{doublespace}
\end{enumerate}
\begin{doublespace}
\noindent 
\begin{figure}
\includegraphics[width=8.75cm]{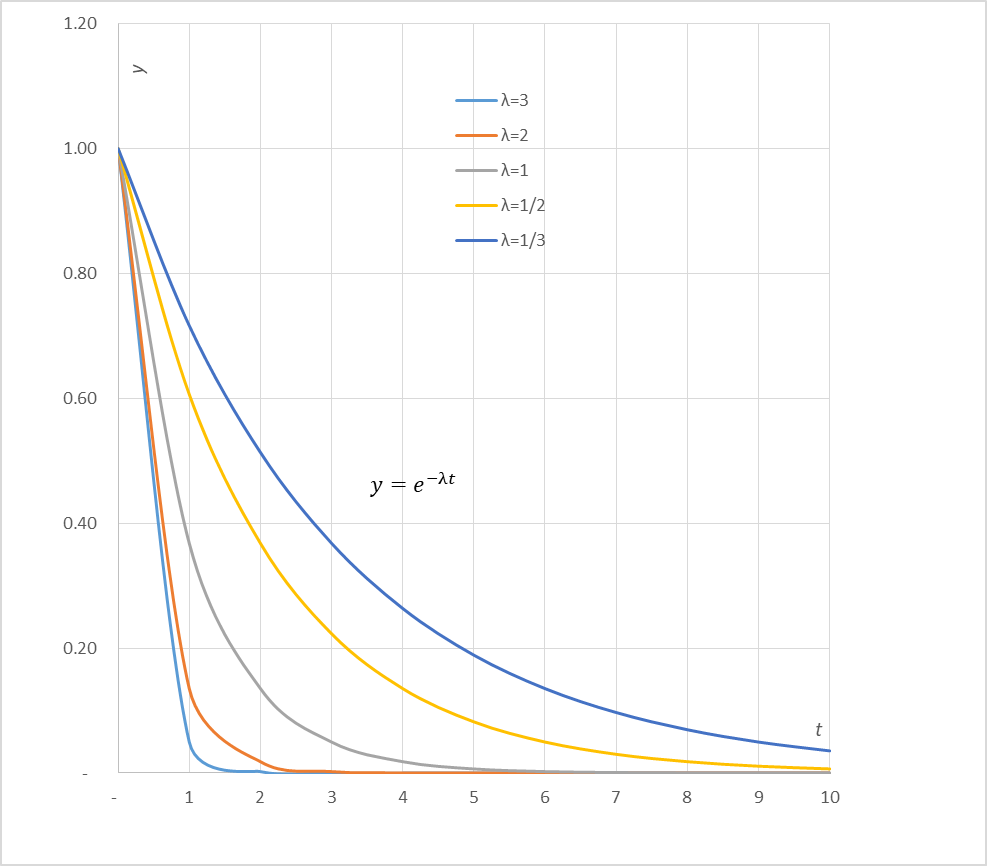}\includegraphics[width=8.75cm]{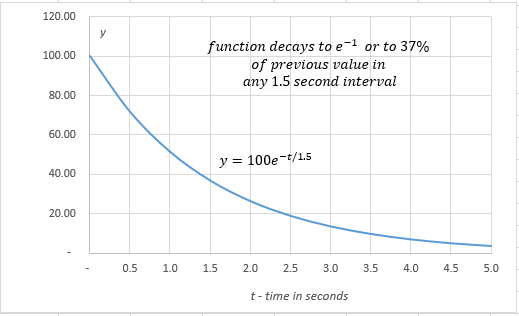}\linebreak{}
\includegraphics[width=17.5cm]{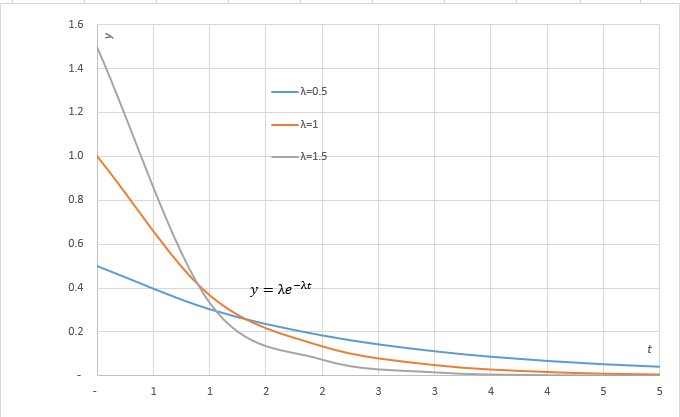}

\caption{\label{fig:Exponential-Decay:-Time}Exponential Decay: Time Value
of Grades}

\end{figure}

\noindent 
\begin{figure}
\includegraphics[width=8.75cm]{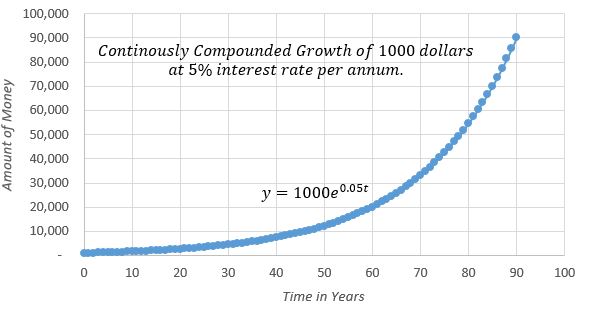}\includegraphics[width=8.75cm]{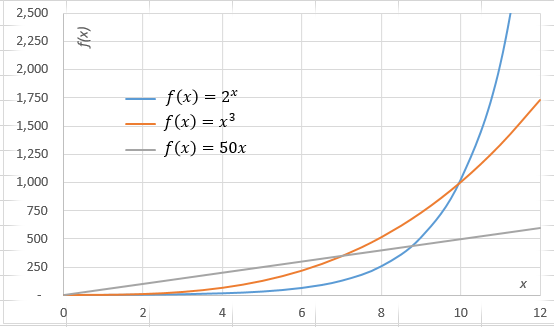}

\caption{\label{fig:Exponential-Growth:-Time}Exponential Growth: Time value
of Money or Friendship}

\end{figure}

\noindent If we have trouble accepting the first two axioms, we must
question our moral compass. If we are judging a person, based on grades
received long back; is that a good criteria we are using for evaluating
someone? If friendships don't grow over time (or, somehow, become
less), should we not examine, whether we as a society, are becoming
too caught up with ourselves, (our lives perhaps) to ignore (or wrong
or begrudge or resent) the ones that have been a part of our lives,
for long periods of time? Hence the NPV (Net Present Value: Ross,
Westerfield \& Jaffe 2002; End-note \ref{Net Present Value}) of any
grade trade is significantly higher for the person who is giving up
the higher grade to make a friend, who is likely to wish him well
for the rest of his life and also get some monetary benefit, which
is going to grow, unless the federal reserve, or other central banks,
intervene to set zero interest rates. The third axiom has countless
papers discussing it, hence we do not consider that in detail here
(Petters \& Dong 2016 has a comprehensive discussion of the time value
of money and all the related paraphernalia; also see: Cochrane 2009;
Bierman Jr \& Smidt 2012).
\end{doublespace}
\begin{doublespace}

\subsection{Some More Simple Lessons from a Seemingly Complex Marketplace}
\end{doublespace}

\begin{doublespace}
\noindent We collect a few more lessons from the financial markets
for our system (which, could be a computer software solution as well),
that is meant for both trading and grading. If we can trade a stock
or financial security hundreds of times in a day (or, perhaps in a
millisecond, since it is becoming harder to keep track of the high
frequency with which, trades happen these days, in certain venues),
we don't have to restrict our grade trade to happen just once. The
grades can (or will?) be traded back; in fact, once the necessity
of the grade for the person who would miss an opportunity, financial
or otherwise, has disappeared, the grades can, (should?) be reversed
back. When the reverse swap happens, we need to make sure, that the
financial component obeys any net present value constraints, or, the
time value of money rules; this would be in comparison to the financial
component in the first trade.

\noindent Also, as with any trading, it would only be proper, to ensure
that some licensing requirements are met. There are many regulatory
bodies in the financial industry, which provide licenses allowing
the practice of different types of regulated activities. Some examples
of these licenses in financial services are: Series 7, SFC Type 4,
and so on (Warschauer 2002; Goetz, Tombs \& Hampton 2005 are discussions
about how colleges can prepare students to get licenses as financial
advisers or practitioners, aiding with an early transition into such
careers; clearly in our case there will need to be involvement from
universities or educational institutions; End-note \ref{List of Securities Examinations};
Egan, Matvos \& Seru 2016 find that misconduct is very prevalent among
financial advisers). The license to trade will be granted only to
students, that can provide proper documentation, establishing their
requirement for maintaining, a certain level of grades, who will be
the buyers in the first round of trades on grades; and to all students
with higher grades, since they would be the sellers in the first round.
When the grades are reversed of course the roles would be exchanged.
To be more precise in the use of financial terminology, we could term
this a grade swap and the mechanics of how it would work, would have
lots of similarities, with an interest rate swap (Hull \& Basu 2016;
Tuckman\& Serrat 2011; Darbyshire 2016; End-note \ref{Interest Rate Swap}).
\end{doublespace}
\begin{doublespace}

\subsection{Luminaries in the Light Pool}
\end{doublespace}

\begin{doublespace}
\noindent It would seem on the surface, before actually engaging in
establishing this market and seeing how this experiment would play
out, that there would many buyers for high grades, but perhaps less
students willing to lower their grades. Hence, in addition to the
two benefits they would receive, the school could provide a certificate
stating that this person was willing to give up something, very precious
in the perception of most, to aid another fellow student, and in the
process, displayed exemplary humanitarian ideals (perhaps, no less
than blood or organ donors, who save the lives of others in need:
Becker \& Elias 2007; Erin \& Harris 2003; Barnett, Blair \& Kaserman
1992 are discussions about markets for human organs).

\noindent It is very likely that this market for trading grades might
raise some ethical questions regarding whether it is right, or, wrong,
to allow trades on grades. But surely, in a world where we allow emissions,
or, pollution to be traded (Drury, Belliveau, Kuhn \& Bansal 1998;
Burtraw, Evans, Krupnick, Palmer \& Toth 2005; Demailly \& Quirion
2008; Fisher-Vanden \& Olmstead 2013), allowing trades on grades,
only when someone badly needs a higher grade to survive and complete
their education, must be an acceptable option. If not, have we completely
lost our sense of justice (and our minds too?) or have we no idea
or clue about what we should or should not be trading?

\noindent Another dilemma is whether, corporations would view this
as a sort of deception being used by universities to send them inferior
candidates. For this we have two refutations: as mentioned, the grades
can be reversed back before graduation; also, employers should not
be overtly concerned about grades. All corporations have extensive
training programs to coach their employees for the duties they would
need to perform. If this becomes such a sticking point, which seems
unlikely, but just for completeness we mention that perhaps, corporations,
also other institutions, need to relax their requirements about using
grades as selection criteria. This is especially important for internships,
since the grades might not yet have been reversed. We consider the
topic of education as a signal of better skills later in section \ref{sec:The-Crux-of}.

\noindent A much more effortless solution than trading grades would,
of course, be to eliminate grading based on curves, eliminate grades
entirely or not have such a strong link between grades and financial
assistance. We could also cultivate an environment where grades are
not taken too seriously, but are viewed simply as the outcome in any
competitive spectator sport, wherein, we expect to have winners and
losers; we simply need to await the next round of games, while reminding
ourselves that winners can take all in some scenarios (Frank \& Cook
2010; Garcia‐del‐Barrio \& Pujol 2007; Leeds \& Kowalewski 2001; Whannel
2005; Frank 1999 discuss the winner take all effect in sports and
other areas of society including higher education applications, which
unfortunately are quite sensitive to rankings. Though this effect
might be okay in some situations, it does cause deep rifts in segments
of society; our core message of ignoring artificial categorizations
is likely to produce a tempering effect). Alternately, we could remind
Mark Zuckerberg, that his honorable ambition of advancing human potential
includes ensuring that humans who have a need are to be taken care
of. Because, if we take care of people they will advance their own
potential and take care of the world (highly likely, that he might
have overlooked this facet of human capability, or, we might have
missed the connection between his dreams and our conundrums).
\end{doublespace}
\begin{doublespace}

\subsection{\label{subsec:Risk-Management-of}Risk Management of Structured Products
\& Grade Payoffs}
\end{doublespace}

\begin{doublespace}
\noindent A picture is worth a thousand words; and a picture with
a few words, must be worth much more? (Figures \ref{fig:Using-Financial-Derivatives},
\ref{fig:Payoff-Diagram-of}). If we really need to have grades, we
suggest one approach to combine, the exam scores, towards a final
grade (better termed, the overall grade, since we might allow students
to retake the exams, so they can improve their scores. They can retake
the exams, when the course is offered next {[}semester or year{]},
so that the overhead, or, the cost of administering the exam, is minimal).
Here, we suggest the use of Structured Products (common in Finance,
Kat 2001) to determine the weight of various component exams towards
the final grade. Introducing a random component to the weights, and
some complex notation, means that perhaps, we will realize that no
matter what the weight, if we study the material well, it will be
reflected in our grades. Also, the one place, many of us, are likely
to notice notation (which is easily overlooked in books, papers and
presentations), is with regards, to how exam weights are calculated.

\noindent If we cannot completely do away with grades, then perhaps
an alternate approach for a grade construction scheme can be,
\end{doublespace}
\begin{enumerate}
\begin{doublespace}
\item Attendance, assignments and participation, should guarantee a C+ in
tough courses. 
\item The mid-term and the finals (plus other components) can determine
whether someone will get A+ or B+. 
\item The percentage weight of the mid-term examination, M\%, to the overall
grade can be given by expression \ref{eq: Mid-Term Weight} below. 
\item The percentage weight of the final examination to the overall grade
can be given by expression \ref{eq:Final Weight} below. 
\item Here, $X$ is the volatility of the mid-term scores, of the exam taking
members. The expressions \ref{eq: Mid-Term Weight} and \ref{eq:Final Weight}
are aimed at ensuring that students pay attention to the principles
and usage of logarithms and min / max functions. This can also be
helpful to understand how payoffs of many financial products can be
made to depend on variables that are not seemingly related to the
financial markets and how creativity is important for the creation
of structured products. It is important to keep in mind that many
alternate expressions can be used to obtain different types of payoffs.
\item It should be clear that the below is a very specific example, wherein
the weight of the mid-term score changes between 0 to 15\%; and the
final score weight changes as the difference between 40\% and the
mid-term weight. The extent to which the score weights will change
depends on the performance of the entire class, more specifically,
it is based on the volatility of the scores (Broadie \& Jain 2008;
Javaheri, Wilmott \& Haug 2004 discuss derivative products that depend
on the volatility or the variance, of certain assets. In this case,
the scores are given to human assets).
\begin{equation}
\min\left(15,\max\left(0,\frac{10^{\left\{ \sqrt{\ln\left(e^{4}\right)}\right\} }}{\left\{ \sqrt{\ln\left(e^{25}\right)}\right\} }-X\right)\right)\label{eq: Mid-Term Weight}
\end{equation}
\begin{equation}
40-\min\left(15,\max\left(0,\frac{10^{\left\{ \sqrt{\ln\left(e^{4}\right)}\right\} }}{\left\{ \sqrt{\ln\left(e^{25}\right)}\right\} }-X\right)\right)\label{eq:Final Weight}
\end{equation}
\item Another example of using random variables as exam or assignment score
weights could be more directly related to finance. If the students
are managing an investment portfolio for an investment analysis course,
then the weights of midterms, assignments or finals could depend on
the volatility of the portfolio value or the P\&L (profit and loss).
Such an approach could mean that the weights might be unique for each
individual or group in the course, leading to the scores being directly
tied to their own destiny or the matter is much more in their own
hands than otherwise. This is equivalent to each individual or group
picking random numbers (surely an alternative for courses where there
is no portfolio being managed), which would then become inputs into
the structured product created to give the weights for their scores.
\end{doublespace}
\end{enumerate}
\begin{doublespace}
\noindent 
\begin{figure}
\includegraphics[width=17.5cm]{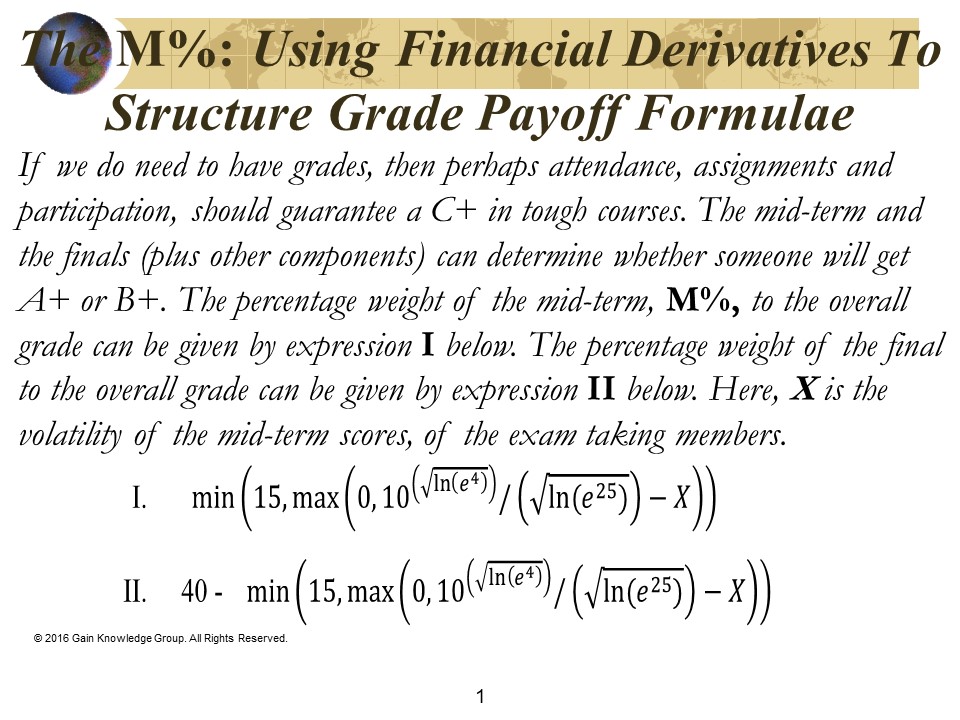}

\caption{\label{fig:Using-Financial-Derivatives}Using Financial Derivatives
to Structure Grade Payoff Formula}
\end{figure}

\end{doublespace}

\begin{figure}
\includegraphics[width=12cm]{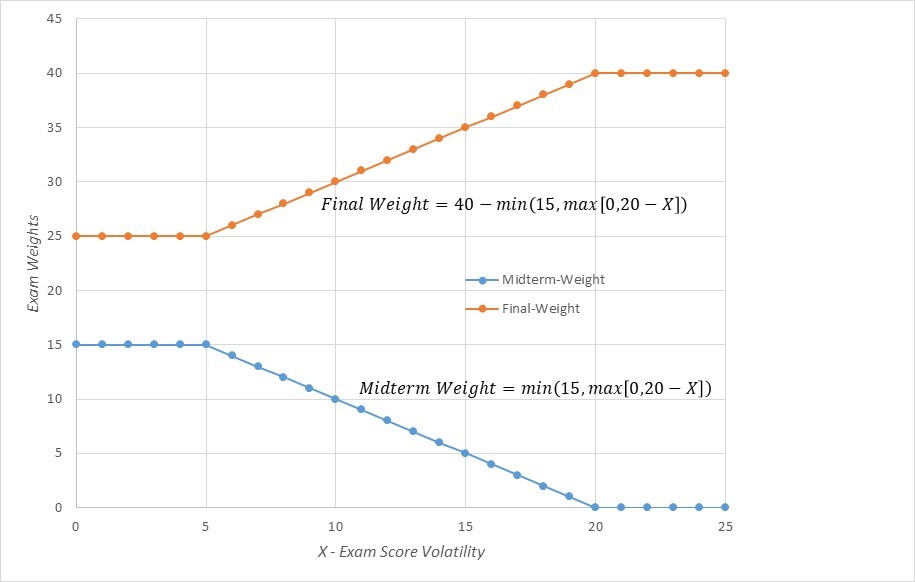}\includegraphics[width=5cm]{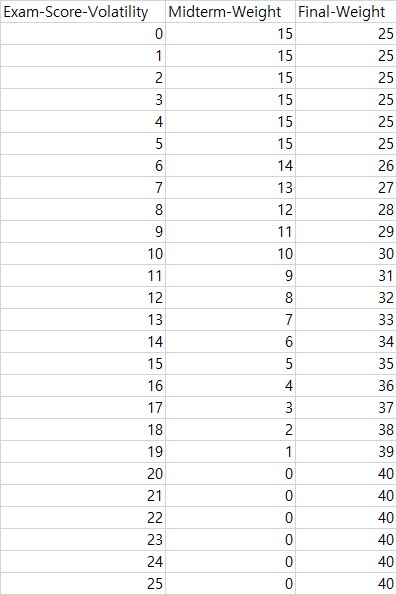}

\caption{\label{fig:Payoff-Diagram-of}Payoff Diagram / Table for Exam Structured
Products in Figure \ref{fig:Using-Financial-Derivatives}}

\end{figure}

\begin{doublespace}

\section{\label{sec:Infinite-Progress-Benchmark}Infinite Progress Benchmark}
\end{doublespace}

\begin{doublespace}
\noindent Another more rigorous mathematical justification to eliminate
grading and other methods of labeling learning is outlined using an
axiom and a theorem that follows from it below.
\end{doublespace}
\begin{ax}
\begin{doublespace}
\noindent The benchmark to be successful in a course is to have made
infinite progress. 
\end{doublespace}
\end{ax}
\begin{doublespace}
\noindent Surely, anyone who has made such a tremendous amount of
progress should not just be deemed to have completed the course, but
to have completed the course with the highest possible distinctions
or with an \textbf{A} grade say. 
\end{doublespace}
\begin{thm}
\begin{doublespace}
\noindent To make a significant amount of progress, in anything we
do, we just need to recognize that anyone who has taken one step forward
has made an infinite percentage change from their starting position,
which we label as zero. It is safe to assume that infinite progress
is significant progress.
\end{doublespace}
\end{thm}
\begin{proof}
\begin{doublespace}
\noindent Since the starting point is zero and the first step denotes
a unit increment, the percentage change becomes, 
\[
\%\text{ Change from Zero to One}=\frac{1-0}{0}=\infty
\]
To be precise, we need to express this as, 
\[
\%\text{ Change from Zero to One}=\underset{a\rightarrow0^{+}}{\lim}\frac{1-a}{a}=+\infty
\]
Here, $a$ approaches zero from the positive real number axis, denoted
as $a\rightarrow0^{+}$. This completes the proof (End-note \ref{Division-by-Zero}).
\end{doublespace}
\end{proof}
\begin{doublespace}
\noindent If we think of any course as a mountain and the enrolled
students as hikers seeking to climb to the summit, then as the course
progresses and a certain amount of time has elapsed, someone would
have climbed higher than others due to innate ability, or due to better
practice either as the course goes on or some familiarity with the
concepts from before the course started. Here the analogy is that
someone who has climbed higher is to be awarded a better grade. We
need to recognize that our real competition is not with anyone else,
but with ourselves. Our real goal from any learning objective (or
any endeavor) is that we are better than our own past selves.

\noindent If we have moved forward compared to where we were when
we started, we have made infinite progress and it deserves its due
recognition. If a student has learned one new concept or idea in a
course, they have fulfilled the criteria for infinite progress. If
this seems like a rather easy condition, we need to remind ourselves
that over the course of an entire formal educational cycle many of
us might not even remember the names of the courses we have taken,
let alone mastering all the material we have studied. If we learn
one concept thoroughly in a course, that can be a noteworthy accomplishment.
Hence, if everyone is making infinite progress, their performance
is equivalent in one particular way and needs to be acknowledged accordingly.
\end{doublespace}
\begin{doublespace}

\section{\label{sec:Turing-Tests-for}Turing Tests for Learning and Teaching}
\end{doublespace}

\begin{doublespace}
\noindent The Turing Test (TT) developed by Alan Turing (Turing 1950;
French 2000 chronicles the comments and controversy surrounding the
first fifty years of the TT; End-note \ref{enu:The-Turing-test}),
is a test of a machine's ability to exhibit intelligent behavior equivalent
to, or indistinguishable from that of a human. Turing proposed that
a human evaluator would judge natural language conversations between
a human and a machine designed to generate human-like responses. The
evaluator would be aware that one of the two partners in conversation
is a machine, and all participants would be separated from one another.
The conversation would be limited to a text-only channel such as a
computer keyboard and screen so the result would not depend on the
machine's ability to render words as speech (Turing originally suggested
a teleprinter, one of the few text-only communication systems available
in 1950). If the evaluator cannot reliably tell the machine from the
human, the machine is said to have passed the test. The test does
not check the ability to give correct answers to questions, only how
closely answers resemble those a human would give.

\noindent In an educational setting, we can devise a Turing Test for
Learning (TTL) similar to the TT, that would be useful for conducting
examinations and evaluating performance. The first phase of assessment
would involve answering a written test based on standard questions
(as the discussion progresses, it should become clear that the questions
should be designed to be more than straightforward application of
concepts, they should gauge the level of comprehension of the topics).
The student can use any resource possible including the internet,
talking to classmates, professors (including the instructor of the
concerned course), books, notes, etc. The students can even leave
the examination room and take the test at any location they deem conducive
for them to answer the written part. They can also be provided ample
time as they deem necessary (in the versions of this test we have
conduced, we gave more than one week for some courses).

\noindent The first phase will contribute only partially towards the
total score. The second phase of the test, which makes up the rest
of the total score, involves the student trying to convince the teacher
how well they have understood the answers they have submitted in writing.
The weighting scheme that combines the first phase and second phase
towards an overall score can even be done based on the discussion
in section \ref{subsec:Risk-Management-of}.

\noindent An often omitted criteria that needs to be considered when
administering the TTL is the ability, or, the level of skill of the
person conducting the test. This gives us another tool for examining
the ability of teachers, which becomes the Turing Test for Teaching
(TTT). Surely, different individuals are satisfied with different
levels of impersonation. When we see any drama, play or movie that
depicts the life of any real person, (while reminding ourselves that
movies might not be real, but real life can become movies); different
people are satisfied with different levels of acting ability. We all
know that the person playing the role in the theatrical version is
not the same, as the person, that is being enacted. But in many cases,
(perhaps, in most cases, when it is well produced), we leave feeling
satisfied with the result of the replication. The lesson for us here
is this: how far does the test administrator need to go, to believe
that the student perfectly understands the subjects being discussed.

\noindent To obtain a lower bound for passing the test, we reason
as follows: our benchmark for success (section \ref{sec:Infinite-Progress-Benchmark})
implies that the benchmark for learning has been surpassed if at-least
one concept has been mastered and the corresponding questions are
answered to the satisfaction of the person administering the TTL.
In this context, answering a question and the relevant decision making
involved are the lessons learnt, at display to satisfy the success
threshold. In an analogous manner, someone can be deemed to have succeed
in the TTT if they have ensured that any student has comprehended
one additional concept compared to when the student started learning
under this instructor.
\end{doublespace}
\begin{doublespace}

\section{\label{sec:The-Crux-of}The Crux of the Curve}
\end{doublespace}

\begin{doublespace}
\noindent While it is tempting to feel contented, with any new innovation,
or, solution that comes up for any problem; the real responsibility
lies in distilling down to the heart of the matter, which in this
case, are the bells and whistles we have attached to the the act of
transferring knowledge. We need to see beyond the artificial ornaments
that we confer upon the torch bearers of knowledge in our later generations.
These are nothing but certificates, diplomas, awards and rankings
of the institutions from which these artifacts have been gathered
that adorn all of us covering the true beauty within each of us. 

\noindent There is a huge literature on the signalling effect of grades
(Grant 2007; Liu \& Neilson 2011; Arkes 1999; Bills 2003). A point
to be noted is that the industry is less critical of lower grades
while higher education institutions are more stringent in checking
the grades of incoming student applications. If we need to rely on
the signals given by others and are unable to apply our independent
thinking to make decisions, we need to question why everyone needs
an education. It is worth highlighting that any recognition we give
to anyone, or, any form of respect for anyone is well placed. But
we should not let these adornments distract us from the luminous brilliance
that can burn brightly from each of us, if we are given the right
spark that can ignite it and a conducive environment that can nurture
it. If we feel that it is tedious to understand each person that we
meet and use our interactions to gauge how to give flight to everyone's
full potential, then perhaps, we all need more training on how to
be better social beings (not just human beings?), than the time we
spend reading textbooks and gathering decorations that merely serve
to act as diversions.

\noindent This is the appreciation we need to give all our fellow
beings and is our primary duty not just to others, but to ourselves,
since we are all creations of the same mother, evolution (Darwin 1859;
Dawkins 1976), that has nourished us for millions of years through
countless trials and errors, making us perfect in every way possible,
tirelessly, cycle upon cycle. Everything we do is because of this
evolutionary training we have received for millions of years (this
hints at an answer to the question: what we should do in terms of
life goals and career choices; perhaps, we should pursue what interests
us rather than what we seem to be good at; simply because with such
extensive preparation we can be good at everything and accomplish
whatever we aspire to do). The information from these historical lessons
are coded and passed on through our genes so that there is improvement
with every generation (International Human Genome Sequencing Consortium
2001; Venter, Adams, Myers, ... \& Gocayne 2001). Again we need to
be reminded that perfection is a moving target. We have perfection
only till the next shortcoming is spotted and then the tireless cycles
that mother nature spins will continue with improvements to adapt
to the new challenges.

\noindent Needless to say, the introduction of any solution, is fraught
to bring new problems of their own, or unintended consequences, if
you will. Periodic revisions of any policy and ensuring that organizations
are not too complex, (maybe smaller would suffice, to keep it less
complex, in many cases?) might be some ways to minimize the unintended
consequences, within the uncertain world, we live, trade and grade
in.
\end{doublespace}
\begin{doublespace}

\section{Conclusion}
\end{doublespace}

\begin{doublespace}
\noindent We have developed a framework and a solution to use the
financial markets to facilitate trades on grades. Students that have
good scores or higher grades, can trade their grade with someone that
might end up losing their scholarships, or other opportunities, because
of missing out the grades they need. In return, the person who is
giving up the higher grades, will get two things: 1) They earn the
good will and friendship of the person who would have lost out a much
needed source of money. 2) They can get a certain percentage of the
funds that might have been lost without the trade.

\noindent We have considered the moral, ethical, legal, humanitarian
constraints and considerations involved in managing such a grading
/ trading system. The other unintentional outcomes, (due to having
man-made furnishings labeling us), that have arisen and distracted
us from properly evaluating, the true splendor of all of nature's
creations, were pointed out. We have barely stopped discriminating
based on race, religion, color, other dividers; we have found new
ways to differentiate, such as the ranking of the schools we have
attended, the grades we obtain, the advanced degrees we have, the
titles we confer upon ourselves both within the corporate culture
or outside, the neighborhoods we live in, and so on.

\noindent The current set up of rankings and the best ranked institutions
attracting the best minds and resources, become self-fulfilling prophecies
that segment society, which is one of the primary outcomes that education
seeks to eradicate (if it does not, then perhaps, it should; Kashyap
2017c has a discussion of one unintended yet welcome consequence of
trying to understand the world, also known as knowledge creation or
research, which is that, we might end up understanding one another
better, becoming more tolerant in the process). Excessive competition
in the process of transferring knowledge dilutes the essence of collaboration,
the inculcation of which, among all us, is the end game, or, the overall
goal of education. Surely, if we need to produce better citizens,
who are willing to lend a helping hand, to their fellow human beings
who are in need, we need to learn to co-operate more than we compete,
since we have plenty of sports, for those of us, who have minds (and
bodies?), made for competitions. On this note, if we start of think
of grades as games, then that might solve the problem too, and we
could use a reminder that, it does not matter who wins, or, loses,
as long as the games go on.

\noindent Our thought experiment becomes another example of ideas
from diverse fields being applied to solve problems in a particular
field illustrating that the many disciplines we have created are artificial.
In this case, finance concepts and the principles of trading were
utilized to solve a problem in education and deeper issues for why
problems such as these arise were discussed. It should be clear that
this experiment applies only to higher education institutions, though
we have not been specific about what higher education means. That
being said, if someone is old enough to elect the head of their nation,
they can be deemed old enough to trade their grade. Though, if we
want allow toddlers in kinder-garden or kids in primary school to
trade grades, we should be more worried about whether the labelling
business has started at such an early age. We will leave those cans
of worms unopened for now.

\noindent The creation of a grade swap market is purely hypothetical.
The ramifications of actually embarking upon such an endeavor are
many, and such an effort is bound to be met with numerous challenges.
A simple suggestion that the thought process in this article brings
about is not to have grades or rankings determine the outcomes in
any selection criteria except have a cutoff point: either for jobs,
higher studies, financial scholarships, etc unless those rankings
are used only for entertainment such as in reality shows or spectator
sports. Instead the one use for grades or rankings might be that they
are helpful to let students or recipients know, what are their areas
of weakness and what improvements might be necessary. All trophies
that we collect (including certificates, grades, medals etc.) should
really be viewed as personal equity or private equity (borrowing another
widely used term in finance). This issue needs to be thought about
by many, and tackled by everyone, before a satisfactory solution can
be put in place.
\end{doublespace}
\begin{doublespace}

\section{Acknowledgments and End-notes}
\end{doublespace}
\begin{enumerate}
\begin{doublespace}
\item Dr. Yong Wang, Dr. Isabel Yan, Dr. Vikas Kakkar, Dr. Fred Kwan, Dr.
William Case, Dr. Srikant Marakani, Dr. Qiang Zhang, Dr. Costel Andonie,
Dr. Jeff Hong, Dr. Guangwu Liu, Dr. Humphrey Tung and Dr. Xu Han at
the City University of Hong Kong provided advice and more importantly
encouragement to explore and where possible apply cross disciplinary
techniques. The students of SolBridge International School of Business
provided the inspiration for this paper and many others. The faculty
members of SolBridge International School of Business, particularly
Dr. Jay-Won Lee, Dr. Kyun-Hwa Kim, Dr. Andrew Isaak and Dr. Hsing
Chia Huang provided patient guidance and valuable suggestions on how
to further this solution and most importantly about the practical
implications of even considering such an approach.
\item \label{Student}A student is primarily a person enrolled in a school
or other educational institution who attends classes in a course to
attain the appropriate level of mastery of a subject under the guidance
of an instructor. Alternately, a student is anyone who applies themselves
to the intensive intellectual engagement with some matter necessary
to master it as part of some practical affair in which such mastery
is basic or decisive. \href{https://en.wikipedia.org/wiki/Student}{Student, Wikipedia Link}
\item \label{Student Financial Aid}Student financial aid is financial support
given to individuals who are furthering their education. Student financial
aid can come in a number of forms, including scholarships, grants,
student loans, and work study programs. Each of these methods of providing
financial support to students has its advantages and drawbacks. \href{https://en.wikipedia.org/wiki/Student_financial_aid}{Student Financial Aid, Wikipedia Link}
\item \label{enu:Grade-inflation-is}Grade inflation is used in two senses:
(1) grading leniency: the awarding of higher grades than students
deserve, which yields a higher average grade given to students (2)
the tendency to award progressively higher academic grades for work
that would have received lower grades in the past. The majority of
the current debates on grade inflation are about the second issue:
\href{https://en.wikipedia.org/wiki/Grade_inflation}{Grade Inflation, Wikipedia Link}
\item \label{enu:An-intelligence-quotient}An intelligence quotient (IQ)
is a total score derived from several standardized tests designed
to assess human intelligence. Scores from intelligence tests are estimates
of intelligence. Unlike, for example, distance and mass, a concrete
measure of intelligence cannot be achieved given the abstract nature
of the concept of \textquotedbl intelligence\textquotedbl{} (End-note
\ref{enu:Intelligence-has-been}). \href{https://en.wikipedia.org/wiki/Intelligence_quotient}{Intelligence Quotient, Wikipedia Link}
(Weinberg 1989; Bartholomew 2004) describe the status of controversies
regarding the definition of intelligence, whether intelligence exists
and, if it does, whether it can be measured, and the relative roles
of genes versus environments in the development of individual differences
in intelligence.
\item \label{enu:Intelligence-has-been}Intelligence has been defined in
many ways, including: the capacity for logic, understanding, self-awareness,
learning, emotional knowledge, reasoning, planning, creativity, and
problem solving. More generally, it can be described as the ability
to perceive or infer information, and to retain it as knowledge to
be applied towards adaptive behaviors within an environment or context.
\href{https://en.wikipedia.org/wiki/Intelligence}{Intelligence, Wikipedia Link}
(Kashyap 2017d) provides an alternate definition of intelligence and
has a deeper discussion on whether intelligence can be increased or
even created in non-human entities.
\item \label{Grade Curve}In education, marking on a curve or grading on
a curve (also referred to as curved grading, bell curving, or using
grading curves) is a method of assigning grades to the students in
a class in such a way as to obtain a pre-specified distribution of
these grades, such as a normal distribution (also called Gaussian
distribution; Rao 1973; End-note \ref{Normal Distribution}). The
term \textquotedbl curve\textquotedbl{} refers to the bell curve,
the graphical representation of the probability density of the normal
distribution, but this method can be used to achieve any desired distribution
of the grades – for example, a uniform distribution (Rao 1973; End-note
\ref{Uniform Distribution}). \href{https://en.wikipedia.org/wiki/Grading_on_a_curve}{Grading on a Curve, Wikipedia Link}
\end{doublespace}

\begin{doublespace}
\noindent For example, if there are five grades in a particular university
course, A, B, C, D, and F, where A is reserved for the top 20\% of
students, B for the next 30\%, C for the next 30\%-40\%, and D or
F for the remaining 10\%-20\%, then scores in the percentile interval
from 0\% to 20\% will receive a grade of D or F, scores from 21\%
to 50\% will receive a grade of C, scores from 51\% to 80\% receive
a grade of B, and scores from 81\% to 100\% will achieve a grade of
A.
\end{doublespace}
\begin{doublespace}
\item \label{Normal Distribution}In probability theory, the normal (or
Gaussian or Gauss or Laplace–Gauss) distribution is a very common
continuous probability distribution. The normal distribution is sometimes
informally called the bell curve. However, many other distributions
are bell-shaped (such as the Cauchy, Student's t, and logistic distributions).
\href{https://en.wikipedia.org/wiki/Normal_distribution}{Normal Distribution, Wikipedia Link}.
The probability density of the normal distribution is:
\begin{equation}
{\displaystyle f(x\mid\mu,\sigma^{2})={\frac{1}{\sqrt{2\pi\sigma^{2}}}}e^{-{\frac{(x-\mu)^{2}}{2\sigma^{2}}}}}
\end{equation}
where $\mu$ is the mean or expectation of the distribution (and also
its median and mode),$\sigma$ is the standard deviation, and $\sigma^{2}$
is the variance.
\item \label{Uniform Distribution}In probability theory and statistics,
the discrete uniform distribution is a symmetric probability distribution
whereby a finite number of values are equally likely to be observed;
every one of $n$ values has equal probability $1/n$. Another way
of saying \textquotedbl discrete uniform distribution\textquotedbl{}
would be \textquotedbl a known, finite number of outcomes equally
likely to happen\textquotedbl .
\item \label{enu:Taleb and Kahneman discuss Trial and Error / IQ Points}
(Ismail 2014) mentions the following quote from Taleb, “Knowledge
gives you a little bit of an edge, but tinkering (trial and error)
is the equivalent of 1,000 IQ points. It is tinkering that allowed
the industrial revolution''. \href{https://www.youtube.com/watch?v=MMBclvY_EMA}{Nassim Taleb and Daniel Kahneman discuss Trial and Error / IQ Points, among other things, at the New York Public Library on Feb 5, 2013.}
\item \label{Grade-Trade-Market}The creation of a grade swap market is
only a thought experiment and is purely hypothetical. The views and
opinions expressed in this article, along with any mistakes, are mine
alone and do not necessarily reflect the official policy or position
of either of my affiliations or any other agency.
\item \label{Dark Pool}In finance, a dark pool (also black pool) is a private
forum for trading securities, derivatives, and other financial instruments.
One of the main advantages for institutional investors in using dark
pools is for buying or selling large blocks of securities without
showing their hand to others and thus avoiding market impact as neither
the size of the trade nor the identity are revealed until some time
after the trade is filled. However, it also means that some market
participants are disadvantaged as they cannot see the orders before
they are executed; prices are agreed upon by participants in the dark
pools, so the market is no longer transparent. \href{https://en.wikipedia.org/wiki/Dark_pool}{Dark Pool, Wikipedia Link}
\item \label{Knowledge Deities}A knowledge deity is a deity in mythology
associated with knowledge, wisdom, or intelligence. The following
link has a list of such deities. \href{https://en.wikipedia.org/wiki/List_of_knowledge_deities}{Knowledge Deities, Wikipedia Link}
\item \label{Knowledge}Knowledge is a familiarity, awareness, or understanding
of someone or something, such as facts, information, descriptions,
or skills, which is acquired through experience or education by perceiving,
discovering, or learning. \href{https://en.wikipedia.org/wiki/Knowledge}{Knowledge, Wikipedia Link}
\item \label{Epistemology}Epistemology (from Greek, epistēmē, meaning 'knowledge',
and logos, meaning 'logical discourse') is the branch of philosophy
concerned with the theory of knowledge. It is the study of the nature
of knowledge, justification, and the rationality of belief. \href{https://en.wikipedia.org/wiki/Epistemology}{Epistemology, Wikipedia Link}
\item \label{Friend Indeed}The idea that false friends will flake and true
friends will reveal themselves as such in times of adversity is ancient.
\href{https://en.wiktionary.org/wiki/a_friend_in_need_is_a_friend_indeed}{Friend Indeed, Wiktionary Link}
\item \label{Zuckberg 99=000025, Human Potential Facebook Link}A letter
to our daughter, Mark Zuckerberg, Wednesday December 2, 2015: in many
ways the world is getting better ... your life should be dramatically
better than ours today. We will do our part to make this happen, not
only because we love you, but also because we have a moral responsibility
to all children in the next generation. ... We will give 99\% of our
Facebook shares -{}- currently about \$45 billion -{}- during our
lives to advance this mission. We know this is a small contribution
compared to all the resources and talents of those already working
on these issues. But we want to do what we can, working alongside
many others. \href{https://www.facebook.com/notes/mark-zuckerberg/a-letter-to-our-daughter/10153375081581634/}{Zuckberg 99\%, Human Potential Facebook Link};
\href{http://philanthropynewsdigest.org/news/zuckerberg-chan-pledge-fortune-to-advance-human-potential-equality}{Zuckerberg 99\%, Human Potential}
\item \label{Zuckerberg 99=000025 News Links}Dec. 1, 2015, San Francisco:
Mark Zuckerberg, the co-founder and chief executive of Facebook, announced
on Tuesday that he and his wife would give 99 percent of their Facebook
shares “during our lives” — holdings currently worth more than \$45
billion — to charitable purposes. The pledge was made in an open letter
to their newborn daughter, Max, who was born about a week ago. \href{https://www.nytimes.com/2015/12/02/technology/mark-zuckerberg-facebook-charity.html}{Zuckerberg 99\%, NY Times Link};
\href{http://www.bbc.com/news/world-us-canada-34978249}{Zuckerberg 99\%, BBC Link};
\href{http://fortune.com/2015/12/02/zuckerberg-charity/}{Zuckerberg 99\%, Fortune Link}
\item \label{Zuckerberg 99=000025 in Legal Entity}Mark Zuckerberg had said
that he was putting 99 per cent of his Facebook shares into the Chan
Zuckerberg initiative, a legal entity dedicated to \textquotedbl advancing
human potential\textquotedbl{} through personalised learning, curing
disease, connecting people and building strong communities. The Independent
and other media reported that this vehicle is not a charity, but a
limited liability company owned and controlled by Zuckerberg. As a
private company, the Chan Zuckerberg Initiative can spend its money
on whatever it wants - including private, profit-generating investment.
\href{http://www.independent.co.uk/news/business/news/zuckerberg-explains-why-he-is-not-giving-99-of-his-facebook-shares-to-charity-a6760291.html}{Zuckerberg 99\% in Legal Entity}
\item \label{enu:The-reason-Rationalizations}The reason we provide these
rationalizations are simply because most academic journals seem to
require complicated mathematics, {[}see, Kashyap 2017a, for a discussion
of whether mathematics is incomprehensibly difficult, or, whether
it is beautifully simple and has been made tremendously convoluted,
unintentionally, of course{]}, enormous amounts of data crunching,
or, abstruse phrases explaining straightforward concepts, to evaluate,
and in many cases even to acknowledge, the contributions in any paper.
This is again an unintended consequence, (Kashyap 2017b, c have more
details), though these extra measures, act as a certain filtering
procedure for quality and are based on the very noble aspiration of
ensuring that the best knowledge bubbles to the surface. That being
said, it might be more efficient if the editors and reviewers watch
out for the most innovative (to be more precise; a certain minimum
level of innovation should suffice, given our subjective preferences
and the complexity that our decisions need to overcome, we might be
limited in our abilities to select the best works: Kashyap 2017c),
or new content and coach the authors on the steps required to create
a publication in their journal. This is illustrated with a simple
example: a hundred lessons on physics will surely be helpful; but
if we substitute one physics lesson for a lesson in biology, or, chemistry,
or, astronomy, that might be more enlightening, and lead to greater
productivity and impact at a later stage. Referring to our observation
in section \ref{sec:Good-Intentions,-Bad}, that the various fields
are artificial boundaries created by us, leads to the simple conclusion
that, learning about a diverse variety of concepts, which are likely
to be dispersed in multiple fields in our present approach to organizing
and enhancing knowledge, and weaving them together to form solutions,
would have a better chance of success in this world we live in.
\item \label{Net Present Value}In finance, the net present value (NPV)
or net present worth (NPW) is the summation of the present (now) value
of a series of present and future cash flows. Because NPV accounts
for the time value of money NPV provides a method for evaluating and
comparing products with cash flows spread over many years, as in loans,
investments, payouts from insurance contracts plus many other applications.
\href{https://en.wikipedia.org/wiki/Net_present_value}{Net Present Value, Wikipedia Link}
\item \label{List of Securities Examinations}The following link has a list
of securities examinations and the organizations that offer them.
\href{https://en.wikipedia.org/wiki/List_of_securities_examinations}{List of Securities Examinations, Wikipedia Link}
\item \label{Interest Rate Swap}In finance, an interest rate swap (IRS)
is an interest rate derivative (IRD). It involves exchange of interest
rates between two parties. In particular it is a linear IRD and one
of the most liquid, benchmark products. An interest rate swap's (IRS's)
effective description is a derivative contract, agreed between two
counterparties, which specifies the nature of an exchange of payments
benchmarked against an interest rate index. \href{https://en.wikipedia.org/wiki/Interest_rate_swap}{Interest Rate Swap, Wikipedia Link}
\end{doublespace}
\item \label{Division-by-Zero}Although division by zero is not defined
for real numbers, limits involving division by a real quantity $x$
which approaches zero may in fact be well-defined. For example, 
\[
\underset{x\rightarrow0}{\lim}\frac{\sin x}{x}=1
\]

Of course, such limits may also approach infinity, 
\[
\underset{x\rightarrow0^{+}}{\lim}\frac{1}{x}=\infty
\]
For a detailed discussion, see: \href{http://mathworld.wolfram.com/DivisionbyZero.html}{Division by Zero, Mathworld Link};
\href{https://en.wikipedia.org/wiki/Division_by_zero}{Division by Zero, Wikipedia Link}.
\begin{doublespace}
\item \label{enu:The-Turing-test}The Turing test, developed by Alan Turing
in 1950, is a test of a machine's ability to exhibit intelligent behavior
equivalent to, or indistinguishable from that of a human. \href{https://en.wikipedia.org/wiki/Turing_test}{Turing Test, Wikipedia Link}
\end{doublespace}
\end{enumerate}
\begin{doublespace}

\section{References}
\end{doublespace}
\begin{enumerate}
\begin{doublespace}
\item Arkes, J. (1999). What do educational credentials signal and why do
employers value credentials?. Economics of Education Review, 18(1),
133-141.
\item Aron, J. (2016). The new philanthropists.
\item Aviles, C. B. (2001). Grading with norm-referenced or criterion-referenced
measurements: To curve or not to curve, that is the question. Social
Work Education, 20(5), 603-608.
\item Barron, K. E., \& Harackiewicz, J. M. (2003). Revisiting the benefits
of performance-approach goals in the college classroom: Exploring
the role of goals in advanced college courses. International Journal
of Educational Research, 39(4), 357-374.
\item Barnett, A. H., Blair, R. D., \& Kaserman, D. L. (1992). Improving
organ donation: compensation versus markets. Inquiry, 372-378.
\item Bartholomew, D. J. (2004). Measuring intelligence: Facts and fallacies.
Cambridge University Press.
\item Beardsley, T. (1995). For whom the bell curve really tolls. Scientific
American, 272(1), 14-17.
\item Becker, G. S., \& Elias, J. J. (2007). Introducing incentives in the
market for live and cadaveric organ donations. Journal of economic
perspectives, 21(3), 3-24.
\item Bell, J., Grekul, J., Lamba, N., Minas, C., \& Harrell, W. A. (1995).
The impact of cost on student helping behavior. The Journal of social
psychology, 135(1), 49-56.
\item Berndt, T. J. (2002). Friendship quality and social development. Current
directions in psychological science, 11(1), 7-10.
\item Betts, J. R., \& Grogger, J. (2003). The impact of grading standards
on student achievement, educational attainment, and entry-level earnings.
Economics of Education Review, 22(4), 343-352.
\item Bierman Jr, H., \& Smidt, S. (2012). The time value of money. In The
Capital Budgeting Decision, Ninth Edition (pp. 29-59). Routledge. 
\item Bills, D. B. (2003). Credentials, signals, and screens: Explaining
the relationship between schooling and job assignment. Review of educational
research, 73(4), 441-449.
\item Bresee, C. W. (1976). On “Grading on the Curve”. The Clearing House,
50(3), 108-110.
\item Broadie, M., \& Jain, A. (2008). Pricing and hedging volatility derivatives.
The Journal of Derivatives, 15(3), 7-24.
\item Brookhart, S. M., Guskey, T. R., Bowers, A. J., McMillan, J. H., Smith,
J. K., Smith, L. F., ... \& Welsh, M. E. (2016). A century of grading
research: Meaning and value in the most common educational measure.
Review of Educational Research, 86(4), 803-848.
\item Burtraw, D., Evans, D. A., Krupnick, A., Palmer, K., \& Toth, R. (2005).
Economics of Pollution Trading for SO2 and NO x. Annu. Rev. Environ.
Resour., 30, 253-289.
\item Calaprice, A. (2000). The expanded quotable Einstein. Princeton, NJ:
Princeton.
\item Carruthers, C. K., \& Özek, U. (2016). Losing HOPE: Financial aid
and the line between college and work. Economics of education review,
53, 1-15.
\item Chan, W., Hao, L., \& Suen, W. (2007). A signaling theory of grade
inflation. International Economic Review, 48(3), 1065-1090.
\item Chua, A. Y., Aricat, R., \& Goh, D. (2017, September). Message content
in the life of rumors: Comparing three rumor types. In Digital Information
Management (ICDIM), 2017 Twelfth International Conference on (pp.
263-268). IEEE.
\item Clark, J., \& McGoey, L. (2016). The black box warning on philanthrocapitalism.
The Lancet, 388(10059), 2457-2459.
\item Cochrane, J. H. (2009). Asset Pricing:(Revised Edition). Princeton
university press.
\item Cohen-Vogel, L., Ingle, W. K., Levine, A. A., \& Spence, M. (2008).
The “spread” of merit-based college aid: Politics, policy consortia,
and interstate competition. Educational Policy, 22(3), 339-362.
\item Cornwell, C., Mustard, D. B., \& Sridhar, D. J. (2006). The enrollment
effects of merit-based financial aid: Evidence from Georgia’s HOPE
program. Journal of Labor Economics, 24(4), 761-786.
\item Cureton, L. W. (1971). The history of grading practices. Measurement
in Education, 2(4), 1-8.
\item Daley, B., \& Green, B. (2014). Market signaling with grades. Journal
of Economic Theory, 151, 114-145.
\item Dancy, J., Sosa, E., \& Steup, M. (Eds.). (2009). A companion to epistemology.
John Wiley \& Sons.
\item Darbyshire, J. H. M. (2016). Pricing and Trading Interest Rate Derivatives:
A Practical Guide to Swaps.
\item Dawkins, R. (1976). The selfish gene. Oxford university press.
\item Darwin, C. (1859). On the Origin of Species by Means of Natural Selection
(PT. 1); Or, the Preservation of Favored Races in the Struggle for
Life. General Books LLC.
\item Demailly, D., \& Quirion, P. (2008). European Emission Trading Scheme
and competitiveness: A case study on the iron and steel industry.
Energy Economics, 30(4), 2009-2027.
\item DeRose, K. (2005). What is epistemology. A brief introduction to the
topic, 20.
\item Domowitz, I., Finkelshteyn, I., \& Yegerman, H. (2008). Cul de sacs
and highways: an optical tour of dark pool trading performance. The
Journal of Trading, 4(1), 16-22.
\item Drury, R. T., Belliveau, M. E., Kuhn, J. S., \& Bansal, S. (1998).
Pollution trading and environmental injustice: Los Angeles' failed
experiment in air quality policy. Duke Envtl. L. \& Pol'y F., 9, 231.
\item Durm, M. W. (1993, September). An A is not an A is not an A: A history
of grading. In The Educational Forum (Vol. 57, No. 3, pp. 294-297).
Taylor \& Francis Group.
\item Egan, M., Matvos, G., \& Seru, A. (2016). The market for financial
adviser misconduct (No. w22050). National Bureau of Economic Research.
\item Eiszler, C. F. (2002). College students' evaluations of teaching and
grade inflation. Research in Higher Education, 43(4), 483-501.
\item Erin, C. A., \& Harris, J. (2003). An ethical market in human organs.
Journal of Medical Ethics, 29(3), 137-138.
\item Fendler, L., \& Muzaffar, I. (2008). The history of the bell curve:
Sorting and the idea of normal. Educational Theory, 58(1), 63-82.
\item Figlio, D. N., \& Lucas, M. E. (2004). Do high grading standards affect
student performance?. Journal of Public Economics, 88(9-10), 1815-1834.
\item Figueroa, A. (2016). Science Is Epistemology. In Rules for Scientific
Research in Economics (pp. 1-14). Palgrave Macmillan, Cham.
\item Finkelstein, I. E. (1913). The marking system in theory and practice
(No. 10). Baltimore: Warwick \& York.
\item Fisher-Vanden, K., \& Olmstead, S. (2013). Moving pollution trading
from air to water: potential, problems, and prognosis. Journal of
Economic Perspectives, 27(1), 147-72.
\item Frank, R. H., \& Cook, P. J. (2010). The winner-take-all society:
Why the few at the top get so much more than the rest of us. Random
House.
\item Frank, R. H. (1999). Higher education: The ultimate winner-take-all
market?.
\item French, R. M. (2000). The Turing Test: the first 50 years. Trends
in cognitive sciences, 4(3), 115-122.
\item French, M. T., Homer, J. F., Popovici, I., \& Robins, P. K. (2015).
What you do in high school matters: High school GPA, educational attainment,
and labor market earnings as a young adult. Eastern Economic Journal,
41(3), 370-386.
\item Ganchev, K., Nevmyvaka, Y., Kearns, M., \& Vaughan, J. W. (2010).
Censored exploration and the dark pool problem. Communications of
the ACM, 53(5), 99-107.
\item Garcia‐del‐Barrio, P., \& Pujol, F. (2007). Hidden monopsony rents
in winner‐take‐all markets—sport and economic contribution of Spanish
soccer players. Managerial and Decision Economics, 28(1), 57-70.
\item Goetz, J. W., Tombs, J. W., \& Hampton, V. L. (2005). Easing college
students' transition into the financial planning profession. Financial
Services Review, 14(3), 231.
\item Grant, D. (2007). Grades as information. Economics of Education Review,
26(2), 201-214.
\item Grant, D., \& Green, W. B. (2013). Grades as incentives. Empirical
Economics, 44(3), 1563-1592.
\item Gross, J. P., Hossler, D., Ziskin, M., \& Berry, M. S. (2015). Institutional
merit-based aid and student departure: A longitudinal analysis. The
Review of Higher Education, 38(2), 221-250.
\item Guskey, T. R. (1994). Making the grade: What benefits students?. Educational
Leadership, 52(2), 14.
\item Guskey, T. R. (2011). Five obstacles to grading reform. Educational
Leadership, 69(3), 16.
\item Hallam, E. M. (1996). Gods and goddesses: a treasury of deities and
tales from world mythology. Pub Overstock Unlimited Inc.
\item Harackiewicz, J. M., Barron, K. E., Carter, S. M., Lehto, A. T., \&
Elliot, A. J. (1997). Predictors and consequences of achievement goals
in the college classroom: Maintaining interest and making the grade.
Journal of Personality and Social psychology, 73(6), 1284.
\item Henry, G. T., \& Rubenstein, R. (2002). Paying for grades: Impact
of merit‐based financial aid on educational quality. Journal of Policy
Analysis and Management, 21(1), 93-109.
\item Henry, G. T., Rubenstein, R., \& Bugler, D. T. (2004). Is HOPE enough?
Impacts of receiving and losing merit-based financial aid. Educational
Policy, 18(5), 686-709.
\item Hetherington, S. C. (2018). Knowledge puzzles: An introduction to
epistemology. Routledge.
\item Hull, J. C., \& Basu, S. (2016). Options, futures, and other derivatives.
Pearson Education India.
\item International Human Genome Sequencing Consortium. (2001). Initial
sequencing and analysis of the human genome. Nature, 409(6822), 860.
\item Ismail, S. (2014). Exponential Organizations: Why new organizations
are ten times better, faster, and cheaper than yours (and what to
do about it). Diversion Books.
\item Jacoby, R., \& Glauberman, N. (Eds.). (1995). The bell curve debate:
History, documents, opinions. New York: Times Books.
\item Javaheri, A., Wilmott, P., \& Haug, E. (2004). GARCH and Volatility
swaps. Quantitative Finance, 4(5), 589-595.
\item Johnson, V. E. (2006). Grade inflation: A crisis in college education.
Springer Science \& Business Media.
\item Johnstone, D. B. (2004). The economics and politics of cost sharing
in higher education: comparative perspectives. Economics of education
review, 23(4), 403-410.
\item Jones, E. B., \& Jackson, J. D. (1990). College grades and labor market
rewards. The Journal of Human Resources, 25(2), 253-266.
\item Kashyap, R. (2017a). Microstructure under the Microscope: Tools to
Survive and Thrive in The Age of (Too Much) Information. The Journal
of Trading, 12(2), 5-27.
\item Kashyap, R. (2017b). Fighting Uncertainty with Uncertainty: A Baby
Step. Theoretical Economics Letters, 7(5), 1431-1452.
\item Kashyap, R. (2017c). The Economics of Enlightenment: Time Value of
Knowledge and the Net Present Value (NPV) of Knowledge Machines. Working
Paper.
\item Kashyap, R. (2017d). Artificial Intelligence: A Child’s Play. Working
Paper.
\item Kat, H. M. (2001). Structured equity derivatives: the definitive guide
to exotic options and structured notes. John Wiley.
\item Kincheloe, J. L., Steinberg, S. R., \& Gresson, A. D. (1997). Measured
Lies: The Bell Curve Examined.
\item Kohn, A. (2002). The dangerous myth of grade inflation. The Chronicle
of Higher Education, 49(11), B7.
\item Kulick, G., \& Wright, R. (2008). The Impact of Grading on the Curve:
A Simulation Analysis. International Journal for the Scholarship of
Teaching and Learning, 2(2), n2.
\item Kvanvig, J. L. (2003). The value of knowledge and the pursuit of understanding.
Cambridge University Press.
\item Leeds, M. A., \& Kowalewski, S. (2001). Winner take all in the NFL:
The effect of the salary cap and free agency on the compensation of
skill position players. Journal of Sports Economics, 2(3), 244-256.
\item Liu, L., \& Neilson, W. S. (2011). High scores but low skills. Economics
of Education Review, 30(3), 507-516.
\item Loury, L. D., \& Garman, D. (1995). College selectivity and earnings.
Journal of labor Economics, 13(2), 289-308.
\item Ludvik, C. (2007). Sarasvatī, Riverine Goddess of Knowledge: From
the Manuscript-carrying Vī\d{n}ā-player to the Weapon-wielding Defender
of the Dharma (Vol. 27). Brill.
\item Ma, C., \& Schapira, M. (2017). The bell curve: Intelligence and class
structure in American life. Macat Library.
\item Mittal, H. (2008). Are You Playing in a Toxic Dark Pool?: A Guide
to Preventing Information Leakage. The Journal of Trading, 3(3), 20-33.
\item Monks, J. (2009). The impact of merit-based financial aid on college
enrollment: A field experiment. Economics of Education Review, 28(1),
99-106.
\item Murray, C., \& Herrnstein, R. (1994). The bell curve. Intelligence
and Class Structure in American Life, New York.
\item Natriello, G. (1987). The impact of evaluation processes on students.
Educational Psychologist, 22(2), 155-175.
\item Petters, A. O., \& Dong, X. (2016). The Time Value of Money. In An
Introduction to Mathematical Finance with Applications (pp. 13-82).
Springer, New York, NY.
\item Pintrich, P. R., \& De Groot, E. V. (1990). Motivational and self-regulated
learning components of classroom academic performance. Journal of
educational psychology, 82(1), 33.
\item Polloway, E. A., Epstein, M. H., Bursuck, W. D., Roderique, T. W.,
McConeghy, J. L., \& Jayanthi, M. (1994). Classroom grading: A national
survey of policies. Remedial and Special Education, 15(3), 162-170.
\item Pritchard, D. (2009). The value of knowledge. The Harvard Review of
Philosophy, 16(1), 86-103.
\item Pritchard, D. (2018). What is this thing called knowledge?. Routledge.
\item Reeves, D. B. (2001). If you hate standards, learn to love the bell
curve. Education Week, 20(39), 38.
\item Rao, C. R. (1973). Linear statistical inference and its applications.
New York: Wiley.
\item Ross, S. A., Westerfield, R. W., \& Jaffe, J. F. (2002). Corporate
Finance.
\item Rumberger, R. W., \& Thomas, S. L. (1993). The economic returns to
college major, quality and performance: A multilevel analysis of recent
graduates. Economics of Education Review, 12(1), 1-19.
\item Sabot, R., \& Wakeman-Linn, J. (1991). Grade inflation and course
choice. Journal of Economic Perspectives, 5(1), 159-170.
\item Schervish, P. G., Davis, S. A., Cosnotti, R. L., \& Rosplock, K. S.
(2016). Solving the Giving Pledge Bottleneck. The Journal of Wealth
Management, 19(1), 23.
\item Schinske, J., \& Tanner, K. (2014). Teaching more by grading less
(or differently). CBE-Life Sciences Education, 13(2), 159-166.
\item Schneider, J., \& Hutt, E. (2014). Making the grade: a history of
the A–F marking scheme. Journal of Curriculum Studies, 46(2), 201-224.
\item Sjoquist, D. L., \& Winters, J. V. (2015). State merit‐based financial
aid programs and college attainment. Journal of Regional Science,
55(3), 364-390.
\item Smallwood, M. L. (1969). An historical study of examinations and grading
systems in early American universities. Harvard University Press.
\item Stan, E. (2012). The role of grades in motivating students to learn.
Procedia-Social and Behavioral Sciences, 69, 1998-2003.
\item Sternberg, R. J. (1995). For whom the bell curve tolls: A review of
The Bell Curve. Psychological Science, 6(5), 257-261.
\item Stiggins, R. J., Frisbie, D. A., \& Griswold, P. A. (1989). Inside
high school grading practices: Building a research agenda. Educational
Measurement: Issues and Practice, 8(2), 5-14.
\item Taylor, G. R. (1971). The bell curve has an ominous ring. The Clearing
House: A Journal of Educational Strategies, Issues and Ideas, 46(2),
119-124.
\item Tomlinson, C. A. (2005). Grading and differentiation: paradox or good
practice?. Theory into practice, 44(3), 262-269.
\item Tuckman, B., \& Serrat, A. (2011). Fixed income securities: tools
for today's markets (Vol. 626). John Wiley \& Sons.
\item Turing, A. M. (1950). Computing machinery and intelligence. Mind,
59(236), 433-460.
\item Turner, P., \& Coulter, C. R. (2001). Dictionary of ancient deities
(p. 170). Oxford University Press.
\item Venter, J. C., Adams, M. D., Myers, E. W., Li, P. W., Mural, R. J.,
Sutton, G. G., ... \& Gocayne, J. D. (2001). The sequence of the human
genome. science, 291(5507), 1304-1351.
\item Warschauer, T. (2002). The role of universities in the development
of the personal financial planning profession. Financial Services
Review, 11(3), 201.
\item Weinberg, R. A. (1989). Intelligence and IQ: Landmark issues and great
debates. American Psychologist, 44(2), 98.
\item Whannel, G. (2005). Winner takes all: competition. In Understanding
television (pp. 114-126). Routledge.
\item Zangenehzadeh, H. (1988). Grade inflation: A way out. The Journal
of Economic Education, 19(3), 217-226.
\item Zeidner, M. (1992). Key facets of classroom grading: A comparison
of teacher and student perspectives. Contemporary Educational Psychology,
17(3), 224-243.
\item Zhang, L., \& Ness, E. C. (2010). Does state merit-based aid stem
brain drain?. Educational Evaluation and Policy Analysis, 32(2), 143-165.
\item Zhu, H. (2014). Do dark pools harm price discovery?. The Review of
Financial Studies, 27(3), 747-789.
\end{doublespace}
\end{enumerate}

\end{document}